\newif\iflong
\newif\ifshort

 \longtrue

\iflong
\else
\shorttrue
\fi

 \documentclass[a4paper,USenglish,cleveref, autoref, thm-restate,envcountsame]{llncs}

\makeatletter
\renewcommand{\@Opargbegintheorem}[4]{   #4\trivlist\item[\hskip\labelsep{#3#2\@thmcounterend}]}
\makeatother

\newcommand\blfootnote[1]{   \begingroup
  \renewcommand\thefootnote{}\footnote{#1}   \addtocounter{footnote}{-1}   \endgroup
}

\usepackage{enumitem}
\usepackage[utf8]{inputenc}
\usepackage{thm-restate}
 \usepackage{amsmath}     \usepackage{amssymb}     \usepackage{mathtools}   \usepackage{microtype}   \usepackage{multirow}    \usepackage{booktabs}     \usepackage[ruled,linesnumbered]{algorithm2e}  \usepackage[usenames,dvipsnames,table]{xcolor}  \usepackage{tikz}
\usetikzlibrary{shapes.geometric,arrows.meta,arrows}
\usetikzlibrary{calc}
\usepackage{nag}        \usepackage{todonotes}  \usepackage{tcolorbox}  \usepackage[hidelinks]{hyperref}     \usepackage{makecell}
\newcommand{\past}{\emph{past}}
\newcommand{\future}{\emph{future}}
\newcommand{\present}{\emph{present}}
\newcommand{\act}{\emph{act}}
\newcommand{\cyc}{\emph{cyc}}
\newcommand{\vis}{\emph{vis}}

\newcommand{\cc}[1]{{\mbox{\textnormal{\textsf{#1}}}}\xspace} 
\newcommand{\NP}{\cc{NP}}
\newcommand{\FPT}{\cc{FPT}}
\newcommand{\XP}{\cc{XP}}
\newcommand{\Wh}{\cc{W[1]}-\textup{hard}\xspace}
\newcommand{\Whness}{\cc{W[1]}-\textup{hardness}\xspace}
\newcommand{\N}{\mathbb{N}}

\newcommand{\III}{\mathbb{I}}

\newcommand{\bigoh}{\mathcal{O}}
\newcommand{\tw}{\textnormal{tw}}

\newcommand{\UF}{\textsc{Unsplittable Flow}\xspace}
\newcommand{\UBP}{\textsc{Unary Bin Packing}\xspace}
\newcommand{\MC}{\textsc{Multi-Colored Clique}\xspace}

\title{The Parameterized Complexity Landscape of \\ the Unsplittable Flow Problem}
\titlerunning{Unsplittable Flow Problem Parameterized}
\author{Robert Ganian \and Mathis Rocton \and Daniel Unterberger}
\institute{Algorithms and Complexity Group, TU Wien, Vienna, Austria}

\usepackage{nameref,cleveref}
\Crefname{splemma}{Lemma}{Lemmas}
\Crefname{sptheorem}{Theorem}{Theorems}
\Crefname{spdefinition}{Definition}{Definitions}
\Crefname{spproperty}{Property}{Properties}
\Crefname{spcorollary}{Corollary}{Corollaries}

\begin{document}

\definecolor{LightGray}{gray}{0.9}
\definecolor{LightCyan}{rgb}{0.88,0.9,1}
\definecolor{Ours}{rgb}{0.25,0.9,0.5}
\definecolor{Known}{rgb}{1,1,0.25}
\definecolor{Par}{rgb}{0.1,0.5,1}
\definecolor{Unb}{rgb}{1,0.5,0.1}
\newcommand{\pa}[1][{\textsc{Par.}}]{\cellcolor{Par}\colorbox{Par}{\makecell{#1}}}
\newcommand{\un}{\cellcolor{Unb} \textsc{Unb.}}
\newcommand{\spacingtab}{~~~~}

\maketitle

\begin{abstract}
We study the well-established problem of finding an optimal routing of unsplittable flows in a graph. While by now there is an extensive body of work targeting the problem on graph classes such as paths and trees, we aim at using the parameterized paradigm to identify its boundaries of tractability on general graphs. We develop novel algorithms and lower bounds which result in a full classification of the parameterized complexity of the problem with respect to natural structural parameterizations for the problem---notably maximum capacity, treewidth, maximum degree, and maximum flow length. 
In particular, we obtain a fixed-parameter algorithm for the problem when parameterized by all four of these parameters, establish \XP-tractability as well as \Whness\ with respect to the former three and latter three parameters, and all remaining cases remain para\NP-hard.
\end{abstract}

\section{Introduction}\label{sec:intro}

 In the \UF problem, we are given an undirected edge-weighted graph $G$ and a set $T$ of \emph{tasks}, where each task consists of a pair of vertex endpoints and two non-negative integers: a \emph{demand} and a \emph{profit}. The goal is to select a subset $T'$ of the tasks and construct a set $P$ of paths such that:
\blfootnote{The authors acknowledge support from the Austrian Science Foundation (FWF, Project 10.55776/Y1329) and the Vienna Science and Technology Fund (WWTF, Project 10.47379/ICT22029). M.\ Rocton also acknowledges support from the \includegraphics[width=0.4cm]{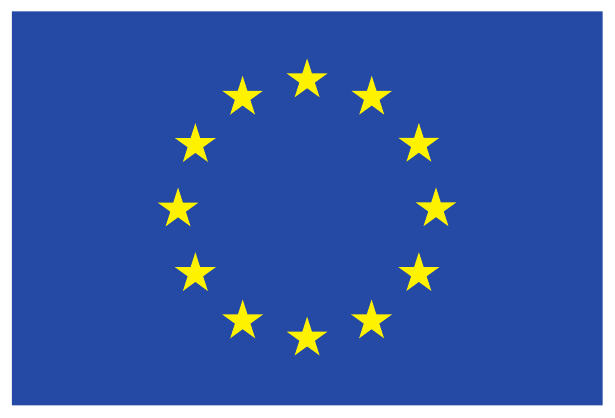} European Union’s Horizon 2020 COFUND programme (LogiCS@TUWien,
No. 101034440).}
\begin{itemize}[topsep=3pt]
\item each task in $T'$ is assigned a path in $P$ which connects its two endpoints;
\item for each edge $e$ in $G$, the total demand of tasks routed via $P$ through $e$ does not exceed the weight of $e$ (i.e., its \emph{capacity}); and
\item among all subsets of tasks and sets of paths satisfying the above two conditions, the sum of the profits of tasks in $T'$ is maximized.
\end{itemize}

\UF has been extensively studied in the literature, and has applications in areas such as resource allocation and scheduling~\cite{PhillipsUW00,CalinescuCKR11}, multi-commodity routing~\cite{ChekuriMS07} and caching~\cite{ChrobakWMX12}. It is also known to be notoriously intractable. First of all, \UF is easily seen to be weakly \NP-hard even when $G$ is restricted to be a $K_2$ via a direct reduction from \textsc{Knapsack}---in other words, if we allow the edge capacities to be encoded in binary, then we cannot hope to even obtain an efficient exact algorithm for a single edge. Naturally, in many cases of interest one need not deal with exceedingly large edge capacities, and so in this article we focus solely on the setting where the capacities are encoded in \textbf{unary}. While the case of $G=K_2$ is trivial in the unary setting, the problem still remains far from tractable: it is known to be strongly \NP-hard on paths even if all tasks have a demand of at most $3$ and all edges have the same capacity~\cite{BonsmaSW11}.

Many of the more recent theoretical works targeting \UF restrict their attention to specific graph classes such as paths~\cite{BonsmaSW14,BatraG0MW15,GrandoniMW017,Wiese17,Anagnostopoulos18a,0001MW21,GrandoniMW22,0001MW22} or trees~\cite{ChekuriEK09,AdamaszekCEW18,Martinez-MunozW21}, and typically deal with the inherent intractability of the problem by aiming for approximation algorithms, parameterized algorithms, or a combination thereof on these graph classes. In this article, we ask a different question: under which conditions can we circumvent the classical intractability of finding an optimal unsplittable flow on \textbf{general graphs}? In particular, under which structural parameterizations of the input can we obtain exact fixed-parameter or at least \XP\ algorithms for \UF?

\smallskip
\noindent \textbf{Contribution.}\quad
Since \UF generalizes the classical \textsc{Edge Disjoint Paths} problem\footnote{\textsc{Edge Disjoint Paths} is equivalent to seeking a routing of all tasks in \UF under the restriction that all demands, capacities and profits are~$1$.}, a natural starting point for our investigation would be to consider the combined parameter of \emph{treewidth} \tw~\cite{RobertsonS86} plus \emph{maximum degree} $\Delta$---indeed, \textsc{Edge Disjoint Paths} is known to be fixed-parameter tractable under the combined parameterization of $\tw+\Delta$~\cite{GanianOR21,GanianO21}. Unfortunately, the aforementioned \NP-hardness of the problem on paths~\cite{BonsmaSW11} immediately rules out even \XP\ algorithms under this parameterization. In combination with the known \NP-hardness of \textsc{Edge Disjoint Paths} on the class $K_{3,n}$ of complete bipartite graphs~\cite{FleszarMS18}, one can rule out tractability under all known structural parameterizations of the input graph alone, including \emph{treedepth}~\cite{sparsity}, the \emph{vertex cover number}~\cite{KorhonenP15,BhoreGMN20,BodlaenderGP23}, \emph{treecut-width}~\cite{MarxWollan14,Wollan15,GanianKS22}, and the \emph{feedback edge number}~\cite{ChaplickGFGRS22,KoanaFN23,BGR24}.

However, even in the considered setting of unary-encoded capacities, the \NP-hardness reduction on paths inherently requires the edge capacities to be sufficiently large (in contrast to the demands, which are small constants). As our first result, we show that if one also parameterizes by the maximum edge capacity in the graph, \UF becomes \XP-tractable:

\begin{restatable}{theorem}{XPctwd}
\label{thm:XPctwd}
\UF\ is in \XP\ parameterized by the maximum capacity $c$, the treewidth $\tw$ and the maximum degree $\Delta$ of the input graph.
\end{restatable}

The proof of the above theorem relies on dynamic programming, where we utilize records that store information about how the selected tasks are routed through the given separator (i.e., bag in the decomposition), whereas each task may of course also visit the bag multiple times. Interestingly, we use trimmed-down records which suppress seemingly crucial information about such routings---in particular the order in which each selected task visits the bag---to achieve a running time with only single-exponential dependency on each of the parameters.

Next, we show that Theorem~\ref{thm:XPctwd} is essentially tight. On one hand, we already noted that one cannot drop $c$ from the parameterization, and similarly one can notice that dropping $\tw$ or $\Delta$ is ruled out by the \NP-hardness of \textsc{Edge Disjoint Paths} on grids~\cite{Marx04} and $K_{3,n}$~\cite{FleszarMS18}. Naturally, the above still leaves open the existence of a fixed-parameter algorithm under the same parameterization. As our second result, we provide a non-trivial reduction from \textsc{Multicolored Clique} which excludes such an algorithm under standard complexity assumptions. Surprisingly, our reduction even rules out fixed-parameter tractability when restricted to the well-studied class of paths.

\begin{restatable}{theorem}{Wctwd}
\label{thm:Wctwd}
\UF\ is \Wh parameterized by the maximum capacity $c$, even when restricted to paths.
\end{restatable}

While the above provides a seemingly complete complexity-theoretic picture of \UF\ under the considered parameterizations (as well as other structural graph parameters; see also the discussion at the beginning of Section~\ref{sec:bounded_l}), the situation is still somewhat unsatisfactory: is there no hope for obtaining a fixed-parameter algorithm by exploiting the structural properties of inputs? In the second part of our article, we identify a restriction on the problem---notably, the maximum length $\ell$ of the paths that can be used to route the tasks---which turns out to yield fixed-parameter tractability of the problem. We remark that such a restriction is far from unnatural: unsplittable flows with a prescribed bound on their length have already been proposed and studied in the literature~\cite{MartensS05,KolmanS06}, and avoiding long routes is critical in several application settings 
 (consider, e.g., the routing of perishable goods or the allocation of bandwidth in a network with communication time/delay constraints).

As our third result, we establish fixed-parameter tractability for \UF\ when additionally parameterized by $\ell$:

\begin{restatable}{theorem}{fptctwdell}
\label{thm:fptctwdell}
\UF\ is fixed-parameter tractable parameterized by the maximum capacity $c$, the treewidth $\tw$, the maximum degree $\Delta$ of the input graph and the maximum length $\ell$ of any admissible flow route.
\end{restatable}

While the proof of Theorem~\ref{thm:fptctwdell} also relies on leaf-to-root dynamic programming (as the vast majority of tree-width based algorithms), it is non-standard in the sense that at each point where we compute new information for our dynamic programming table, we need to invoke a separate dynamic-programming subprocedure as a second layer of the computation. Moreover, the algorithm developed for the proof of Theorem~\ref{thm:fptctwdell} also establishes \XP-tractability for the problem when one drops $c$ in the parameterization: 
 
\begin{restatable}{theorem}{xptwdell}
\label{thm:xptwdell}
\UF\ is in \XP\ parameterized by the treewidth $\tw$ and the maximum degree $\Delta$ of the input graph along with the maximum length $\ell$ of any admissible flow route.
\end{restatable}

The above shows that $\ell$ can essentially ``replace'' $c$ in the parameterization of Theorem~\ref{thm:XPctwd} in order to circumvent the \NP-hardness of \UF. On the other hand, 
 the problem remains \NP-hard even for fixed values of $c+\tw+\ell$ (due to the aforementioned \NP-hardness of \textsc{Edge Disjoint Paths} on $K_{3,n}$~\cite{FleszarMS18}) and also for fixed values of $c+\Delta+\ell$ (due to a different, very recent \NP-hardness proof for \textsc{Edge Disjoint Paths} on grids with bounded-length paths~\cite{EibenGK23}). Finally, we complement Theorem~\ref{thm:xptwdell} with a straightforward reduction from \textsc{Unary Bin Packing}~\cite{JansenKMS10} that rules out fixed-parameter tractability under the same parameterization as Theorem~\ref{thm:xptwdell}. 

\begin{restatable}{theorem}{Wtwdell}
\label{thm:Wtwdell}
\UF\ is \Wh\ parameterized by the treewidth $\tw$ and the maximum degree $\Delta$ of the input graph along with the maximum length $\ell$ of any admissible flow route.
\end{restatable}

Our results are summarized in Figure~\ref{fig:landscape}.
 
\begin{figure}[ht]
\begin{center}
    \begin{tikzpicture}
      
    \node at (0,-0.5) {$\tw+\Delta+c+\ell$};
    
    \node at (-4,-2) {$\tw+\Delta+c$};
    \node at (-1.4,-2) {$\tw+\Delta+\ell$};
    \node at (1.4,-2) {$\tw+c+\ell$};
    \node at (4,-2) {$\Delta+c+\ell$};

    \node at (-2.8,-3.75) {$\tw+\Delta$};
    \node at (-0.9,-3.75) {$\tw+c$};
    \node at (0.9,-3.75) {$\tw+\ell$};
    \node at (2.8,-3.75) {$\Delta+c$};
    \node at (-4.7,-3.75) {$\Delta+\ell$};
    \node at (4.7,-3.75) {$c+\ell$};

    \node at (-4,-5.25) {$\tw$};
    \node at (-1.2,-5.25) {$\Delta$};
    \node at (1.2,-5.25) {$c$};
    \node at (4,-5.25) {$\ell$};

    \node at (0,0.1) {\bf{\large FPT}};
    \node at (-4.2,-0.6) {\bf{\large XP and}};
        \node at (-4.2,-1.2) {\bf{\large W[1]-hard}};
    \node at (4.1,-1) {\bf{\large paraNP-hard}};

    \filldraw[fill = red, draw = white, fill opacity = 0.2] (-6,0.5) to (-6, -6.25) to (6, -6.25) to (6,0.5) to (4,0.5) to [bend left = 25] (0,-1.49) to [bend left = 30] (-1.5,-2.75) to [bend left = 15, rounded corners] (-5,-2.75) to [bend left = 60] (-4,0.5) to (-6, 0.5);
    
    \filldraw[fill = blue, draw = black, fill opacity = 0.2] (-4,0.5) to [bend right = 25] (0,-1.49) to [bend right = 25] (4,0.5) to (-4,0.5);

    \filldraw[fill = yellow, draw = black, fill opacity = 0.2] (0,-1.49) to [bend left = 30] (-1.5,-2.75) to [bend left = 15, rounded corners] (-5,-2.75) to [bend left = 60] (-4,0.5) to [bend right = 25] (0,-1.49);
    \end{tikzpicture}
\end{center} 
\vspace{-0.2cm}
\caption{The complexity landscape of \UF\ under structural parameterizations. Here $\tw$, $\Delta$, $c$ and $\ell$ denote the treewidth, maximum degree, maximum capacity and maximum length of an admissible route, respectively. A discussion of the (non-)applicability of other major structural parameters is provided in Section~\ref{sec:bounded_l}.}
\label{fig:landscape}
\end{figure}
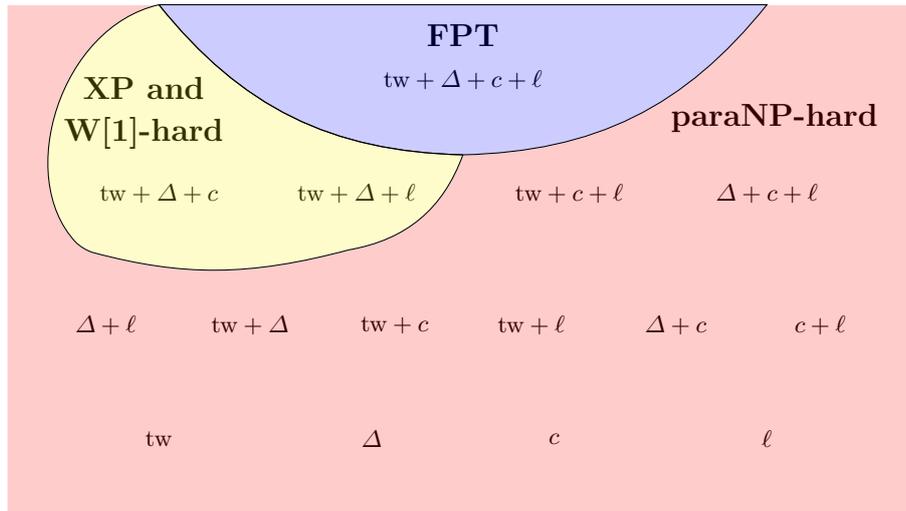

\smallskip
\noindent \textbf{Related Work.}\quad
As mentioned earlier, much of the previous work on \UF\ targeted approximate solutions on paths and trees. Chrobake, Woeginger, Makino and Xu showed that the problem is strongly \NP-hard even on paths with uniform profits and capacities~\cite{ChrobakWMX12}. Anagnostopoulos, Grandoni, Leonardi and Wiese obtained a PTAS for \UF\ on paths~\cite{Anagnostopoulos18a}, and this was later improved to a $(1 + \varepsilon)$ PTAS by Grandoni, Momke and Wiese~\cite{0001MW22}.

\UF on trees is a generalization of the problem of finding integral multi-commodity flows on trees, which is already APX-hard~\cite{GargVY97}. Mart{\'{\i}}nez{-}Mu{\~{n}}oz and Wiese obtained exact and approximate fixed-parameter algorithms for the problem on trees when the solution size is included in the parameterization~\cite{Martinez-MunozW21}. Chekuri, Ene and Korula showed that \UF on trees admits a $\bigoh(\log n)$ approximation when all weights are identical, and $O(\log^2 n)$ otherwise~\cite{ChekuriEK09}.
       
On general graphs, Guruswami, Khanna, Rajaraman, Shepherd, and Yannakakis established the \NP-hardness of approximating \UF\ within a factor of $m^{\frac{1}{2} - \varepsilon}$ for any $\varepsilon>0$~\cite{GuruswamiKRSY03}. 
Approximation algorithms for \UF\ on general graphs have also been studied, e.g., in the earlier works of Kolman and Scheideler~\cite{KS01,KolmanS06} or Baveja and Srinivasan~\cite{BavejaS00}.

\section{Preliminaries}\label{sec:prelim}
We use standard graph terminology~\cite{diestel} and assume familiarity with the foundations of parameterized complexity theory, including the complexity classes \FPT, \XP and \cc{W[1]}, and parameterized reductions~\cite{DowneyF13,CyganFKLMPPS15}. All graphs considered in this work are simple and undirected, and $\N$ is understood to refer to the set of non-negative integers. 
For an integer $n$, we denote by $[n]$ the set $\{1,\dots,n\}$ and set $[n]_0=[n]\cup \{0\}$.
For brevity, we sometimes use $V(G)$ and $E(G)$ to denote the vertex and edge sets of a graph $G$. We use the standard notation of $K_\bullet$ and $K_{\bullet, \circ}$ for complete and complete bipartite graphs, respectively.

\smallskip
\noindent \textbf{Treewidth.}\quad A \emph{tree-decomposition}~$\mathcal{D}$ of a graph $G=(V,E)$ is a pair 
$(D,\chi)$, where $D$ is a tree (whose vertices we call \emph{nodes}) rooted at a node $r$ and $\chi$ is a function that assigns each node $t$ a set $\chi(t) \subseteq V$ such that the following holds: 
 \begin{itemize}[topsep=3pt,noitemsep]
	\item For every $uv \in E$ there is a node	$t$ such that $u,v\in \chi(t)$.
	\item For every vertex $v \in V$, the set of nodes $t$ satisfying $v\in \chi(t)$ forms a nonempty subtree of~$T$.
\end{itemize}

A tree-decomposition is \emph{nice} if the following two conditions are also satisfied:
\begin{itemize}[topsep=3pt,noitemsep]
	\item $|\chi(\ell)|=1$ for every leaf $\ell$ of $T$ and $|\chi(r)|=0$.
	\item There are only three kinds of non-leaf nodes in $T$:
	\begin{itemize}[topsep=3pt,noitemsep]
        \item \textbf{Introduce node:} a node $t$ with exactly
          one child $t'$ such that $\chi(t)=\chi(t')\cup
          \{v\}$ for some vertex $v\not\in \chi(t')$.
        \item \textbf{Forget node:} a node $t$ with exactly
          one child $t'$ such that $\chi(t)=\chi(t')\setminus
          \{v\}$ for some vertex $v\in \chi(t')$.
        \item \textbf{Join node:} a node $t$ with two children $t_1$,
          $t_2$ where $\chi(t)=\chi(t_1)=\chi(t_2)$.
	\end{itemize}
\end{itemize}
The set $\chi(t)$ for a node $t$ of $D$ is called a \emph{bag}. 
The \emph{width} of a tree-decomposition $(D,\chi)$ is the size of its largest bag $D$ minus~$1$, and the \emph{treewidth} of the graph $G$,
denoted $\tw(G)$, is the minimum width of a tree-decomposition of~$G$. 
It is known that a tree-decomposition can be transformed into a nice tree-decomposition of the same width in linear time. 
Efficient fixed-parameter algorithms are known for computing a nice tree-decomposition of near-optimal width~\cite{Kloks94,BodlaenderDDFLP16,Korhonen21}:

\begin{proposition}[\cite{Korhonen21}]\label{fact:findtw} 	There exists an algorithm which, given an $n$-vertex graph $G$ and an integer~$k$, in time $2^{\bigoh(k)}\cdot n$ either outputs a tree-decomposition of $G$ of width at most $2k+1$ and $\bigoh(n)$ nodes, or determines that $\tw(G)>k$.
\end{proposition}  

We let $D_t$ denote the subtree of $D$ rooted at a node $t$, and we use $\chi(D_t)$ to denote the set $\bigcup_{t'\in V(D_t)}\chi(t')$. In the context of leaf-to-root dynamic programming, it will be useful to also define
$\past(t)=\chi(D_t)\setminus \chi(t)$ and $\future(t)=V\setminus\past(t)$; to avoid any confusion, we remark that the latter set also contains $\chi(t)$. We also define the edge set $\present(t)$ containing all edge with precisely one endpoint in $\past(t)$. These will perform the role of the ``boundary'' in all of our dynamic programming algorithms, and we observe that $|\present(t)|$ is upper-bounded by a product of $|\chi(t)|$ and the maximum degree $\Delta$ of the input graph.

\smallskip
\noindent \textbf{Problem Formulation.}\quad
For complexity-theoretic reasons and without loss of generality, in this article we formalize \UF\ as a decision problem as opposed to the optimization variant typically considered in works focusing on approximation. To avoid any confusion, we remark that each of the developed algorithms is exact, deterministic and constructive in the sense that it can also output a suitable witness for every yes-instance---and hence fully applicable also to the optimization variant defined in the introduction.

An instance $\III$ of the \UF\ problem studied here consists of: 
\begin{enumerate}[topsep=4pt]
\item a simple undirected graph $G=(V,E)$, 
\item an edge labeling function $\gamma: E \rightarrow \N$, 
\item a set $\mathcal{T}$ of \emph{tasks} $T\subseteq V\times V \times \N \times \N$, 
\item a \emph{profit target} $\tau$, and
\item an integer length bound $\ell$. 
\end{enumerate} 

We call $\gamma(e)$ the \emph{capacity} of the edge $e$, and for a task $p=(s,t,d,w) \in T$ we call $d=d(p)$ its demand and $w=w(p)$ its profit.
 Given a subset $T'\subseteq T$ of tasks, we say that a \emph{flow routing} is a mapping which assigns to each task $(s,t,d,w)$ in $T'$ an $s$-$t$ path of length at most $\ell$ such that for each edge $e\in E$, the total demand of all paths in the flow routing passing through $e$ is at most $\gamma(e)$. The \emph{profit} of $T'$ is simply the sum of the profits of the tasks in $T'$, i.e., $\sum_{(s,t,d,w)\in T'}w$. The \UF\ problem then asks us to decide whether there exists a subset $T'$ which admits a flow routing and achieves a profit of at least $\tau$. 

We note that $\ell$ is not present in most definitions of \UF\ in the literature, and indeed our first set of results---in Section~\ref{sec:unbounded_l}---provide lower bounds for the case where $\ell$ is ignored (i.e., set to a sufficiently large bound such as $\ell=|V|$) along with algorithms that can deal with any choice of $\ell$. The size $|\III|$ of an instance $\III$ is simply the sum of the sizes of its components, whereas we assume throughout this work that the capacity function $\gamma$ is encoded in unary. This condition may be reformulated as requiring that $|\III|$ upper-bounds the maximum capacity of an edge, and is necessary in the sense that \UF\ with binary-encoded capacities is trivially \NP-hard even on a $K_2$ (see Section~\ref{sec:intro}).

\section{Solving Unsplittable Flow Regardless of Flow Length} \label{sec:unbounded_l}

This section is dedicated to the complexity of \UF\ when the bound $\ell$ on the length of the flow routes is not part of the parameterization. As our first result, we prove Theorem~\ref{thm:XPctwd}, which is the only case that is not para\NP-hard in the part of our landscape that excludes $\ell$ (cf.\ Figure~\ref{fig:landscape}).

\XPctwd*
   
\begin{proof}
We first invoke Proposition~\ref{fact:findtw} to obtain a nice tree-decomposition $\mathcal{D}=(D,\chi)$ of $G=(V,E)$ of width $k$, where $k\leq 2\tw+1$.
Our proof utilizes a dynamic programming routine that proceeds in a leaf-to-root fashion along $\mathcal{D}$.

We begin by defining the syntax of the records $R(t)$ we will use for each node $t$. $R(t)$ contains a set of tuples of the form $(\Lambda, \Theta,\Omega)$, where:
\begin{itemize}[topsep=3pt]
 \item $\Lambda: \present(t) \rightarrow \{ S\subseteq \mathcal{T}, \sum_{t\in S}d(t) \leq c\}$,
\item $\Theta: \big(\bigcup_{e\in \present(t)} \Lambda(e) \big)\rightarrow [\ell]$, and
\item $\Omega \in \N$.
\end{itemize} 

Intuitively, these records capture the following information about a hypothetical solution: $\Lambda$ specifies which tasks are routed through $\present(t)$, $\Theta$ captures their total flow length in $\past(t)$, and $\Omega$ contains information about their profit. We explicitly remark that, (1) unlike one would expect in a similar dynamic program, $\Lambda$ does not store the order in which the selected tasks cross the boundary $\present(t)$, and (2) $\Omega$ will not be formalized as a simple sum over the profits of all routed tasks but rather a ``weighted sum''. To simplify our later operations involving $\Theta$, we use $\Lambda(t)$ as shorthand for $\big(\bigcup_{e\in \present(t)} \Lambda(e) \big)$, i.e., the set of all tasks mapped by $\Lambda$ to at least one edge in $\present(t)$.

In order to formalize the above intuition by defining the semantics of the records, we first need to introduce the notion of a partial solution. A \emph{partial solution} at a node $t$ is a mapping $\zeta$ from a subset $Z\subseteq \mathcal{T}$ of tasks to a multiset $\mathcal{P}$ of sets of paths---each path being composed entirely of edges with at least one endpoint in $\past(t)$---with the following properties:
\begin{itemize}[topsep=3pt]
\item for each $z\in Z$ such that both endpoints lie in the future (i.e., in $\future(t)$), $z$ is mapped to a (possibly empty) set of vertex disjoint paths whose extremities are the vertices of $\chi(t)$.
\item for each $z\in Z$ such that exactly one endpoint $u$ lies in the past (i.e., in $\past(t)$), $z$ is mapped to a set of vertex disjoint paths, where exactly one of these connects $u$ to $\chi(t)$ and all of the (possibly $0$) others have their extremities in $\chi(t)$.
\item for each $z\in Z$ such that both of its endpoints $u$ and $v$ lie in the past, $z$ is either mapped to a single path between $u$ and $v$, or to a set of vertex disjoint paths such that one of these connects $u$ to $\chi(t)$, another one connects $v$ to $\chi(t)$, and all of the (possibly $0$) others have their extremities in $\chi(t)$.
\item for each edge $e\in E(G)$, its capacity $\gamma(e)$ is sufficient to accomodate the demand of all tasks of $Z$ containing a path using $e$.
\end{itemize} 

For a given partial solution $\zeta:Z \rightarrow \mathcal{P}$, we define its \emph{partial profit} $\omega$ as the sum of the profits of all tasks in $Z$ with both endpoints in $\past(t)$ plus half of the profits of all tasks in $Z$ with exactly one endpoint in $\past(t)$. 
A task is \emph{active} in $\zeta$ if it is routed to at least one edge in $\present(t)$, and we note that there can be at most $|\present(t)|\cdot c \leq \tw \cdot \Delta \cdot c$ active tasks.
We denote the set of all active tasks in $\zeta$ as $Z_{\act}$.
For each $z\in Z$, we let $\theta_z$ be the total number of edges in $\zeta(z)$.

We are now ready to formalize the semantics of our records. A tuple $(\Lambda, \Theta, \Omega)$ is \emph{admissible} for $R(t)$ if and only if 
there exists a corresponding partial solution $\zeta: Z\rightarrow \mathcal{P}$ which
\begin{enumerate}[topsep=4pt]
\item achieves a partial profit of $\Omega$, 
\item satisfies $\theta_z=\Theta(z)$ for each task $z\in Z_{\act}$, and
\item for each edge $e\in \present(t)$ and task $z$ in $Z$, $z \in \Lambda(e)$ if and only if $e$ occurs on a path in $\zeta(z)$.
\end{enumerate}
While being admissible is the main prerequisite for being part of our records, we also need to impose two additional conditions which will be important when dealing with join nodes. For a vertex $v\in \chi(t)$, a mapping $\Lambda$ as above and a task $z$, we say that $v$ is a \emph{jump-vertex} if $v$ has two incident present edges which are both mapped by $\Lambda$ to $z$. We say that a tuple $(\Lambda', \Theta', \Omega)$ \emph{directly supersedes} a tuple $(\Lambda, \Theta, \Omega)$ if there is a task $z$ such that (a) $\Lambda'$ can be obtained from $\Lambda$ by removing $z$ from the edges incident to some set of jump-vertices, and (b) $\Theta'$ can be obtained from $\Theta$ by reducing the value of $\Theta(z)$ by a positive integer. We then define the relation of \emph{superseding} as the transitive closure of \emph{directly superseding}, that is, $(\Lambda', \Theta', \Omega)$ \emph{supersedes} $(\Lambda, \Theta, \Omega)$ if the above consideration can be repeatedly applied to obtain the former from the latter.

We are now ready to define our records: a tuple $(\Lambda, \Theta, \Omega)$ is in $R(t)$ if and only if it is admissible and moreover satisfies the following two Optimality Conditions: 
\begin{enumerate}[topsep=4pt]
\item[(i)]  there is no admissible tuple of the form $(\Lambda, \Theta, \Omega')$ such that $\Omega'>\Omega$, and
\item[(ii)] there is no admissible tuple that supersedes $(\Lambda, \Theta, \Omega)$. 
  \end{enumerate}
 In other words, the optimality conditions ensure that we do not keep solutions which~(i) are clearly sub-optimal or (ii) can achieve the same connections outside of $\past(t)$ without requiring longer paths.

Having defined our records, we observe that the total size of $R(t)$ is upper-bounded by $|\III|^{\bigoh(k\cdot \Delta\cdot c)}$ simply due to the syntax: for each of the at most $k\cdot \Delta$ edges in $\present(t)$, there are at most $|\III|^c$ choices of tasks for $\Lambda$ and at most $|\III|$ choices for $\Theta$. Moreover, since the bag of the root $r$ of $\mathcal{D}$ is empty, we have that $\present(r)=\emptyset$; hence $R(r)$ contains a single entry of the form $(\emptyset,\emptyset,\Omega_r)$ where $\Omega_r$ is the maximum profit that can be achieved in $\III$. This means that to solve our problem, it now suffices to show how to compute $R(t)$ for each node $t$ in $\mathcal{D}$. 

The computation will depend on the kind of node $t$ is, and is carried out as follows.

\smallskip
\noindent \textbf{$t$ is a leaf node.} \quad
If $t$ is a leaf node, we recall that $\present(t)=\emptyset$ and hence $R(t)=\{(\emptyset,0,\emptyset\}$.

\smallskip
\noindent \textbf{$t$ is an introduce node.}\quad
Let $t'$ be the unique child of $t$, and notice that $\present(t)=\present(t')$ and $\past(t)=\past(t')$. Thus, partial solutions at the node $t$ are exactly partial solutions at $t'$, and we can correctly set $R(t):=R(t')$.

\smallskip
\noindent \textbf{$t$ is a forget node.}\quad
Let $t'$ be the unique child of $t$ and $u$ be the unique vertex in $\chi(t')\setminus \chi(t)$. In other words, $\past(t)=\past(t')\cup\{u\}$. 

\paragraph*{Construction.}
Before we construct the entries of $R(t)$, we first construct a set of candidate tuples as follows. 
First, for each entry $(\Lambda', \Theta', \Omega')\in R(t')$, we branch over all possible mappings of the ``new'' edges in $\present(t)\setminus\present(t')$ to sets of at most $c$ tasks to construct a set of ``extended mappings' of the form $\Lambda^*:\present(t) \cup \present({t'})\rightarrow \{ S\subseteq \mathcal{T}, \sum_{t\in S}d(t) \leq c\}$. For each such choice of~$\Lambda^*$, we check that request flows are preserved at $u$: for each task $z$ that is routed through (i.e., in the image of) an edge incident to $u$ via $\Lambda^*$, it must either be routed through exactly $2$ such edges if $u$ is not an endpoint of $z$, or exactly $1$ such edge if $u$ is an endpoint of $z$. If this verification fails, we discard the given degenerate choice of $\Lambda^*$.

Next, we update $\Omega'$ with the profit achieved at $u$: for every task that has $u$ as an endpoint and is routed by $\Lambda^*$ through one of the edges incident to $u$, we add one half of the task's profit to $\Omega'$ to obtain the new profit $\Omega$.
We proceed by creating the mapping $\Lambda$ by restricting $\Lambda^*$ to $\present(t)$.

It now remains to construct $\Theta: \Lambda(t) \rightarrow [\ell]$. For a task $z$, let $\delta_u(z)$ denote the number of ``new edges'' incident to $u$ that $z$ is routed through, i.e., $\delta_u(z)=|\{v\in \chi(t)~|~ \in \Lambda((u,v))\}|$. For each $z\in \Lambda(t')$, we now set $\Theta(z):=\Theta'(z) + \delta_u(z)$, while for each $z\in \Lambda(t)\setminus \Lambda(t')$ we simply set $\Theta(z):= \delta_u(z)$. We perform two final checks to discard dead branches: 
\begin{itemize}[topsep=3pt]
\item if $\Theta(z)> \ell$ for any $z$, we discard the tuple (as it does not satisfy the syntax of the records), and
\item for each $z\in \Lambda(t)$ that has at least one endpoint in $\past(t')$, we check that $z\in \Lambda(t')$ and discard the tuple if this fails (since every partial solution would need to route $z$ via $\present(t')$ in order to reach $\present(t)$).
\end{itemize}

This completes the construction of the candidate tuples for $t$. To compute $R(t)$, we now discard every candidate tuple that does not satisfy Optimality Conditions~(i) and~(ii), i.e., we discard tuples which are superseded by another candidate tuple or achieve suboptimal profit. Next, we argue that the records $R(t)$ constructed in this way are correct.

\paragraph*{Correctness.}
Every partial solution $\zeta:Z\rightarrow \mathcal{P}$ at $t$ induces precisely one partial solution $\zeta':Z' \rightarrow \mathcal{P}'$ at $t'$ that is obtained by omitting the edges in $\present(t)\setminus \present(t')$.
        Analogously, every partial solution $\zeta$ at $t$ can be obtained by prolonging and merging some paths in a partial solution $\zeta'$ of $t'$ along the edges between $u$ and $\chi(t)$, as well as creating some new paths on these same edges for newly active tasks. This is exactly equivalent to extending the mapping $\Lambda'$ on $\present({t'})$ to a mapping $\Lambda^*$ defined on $\present(t) \cup \present({t'})$. Note that the partial profit for $(Z, \mathcal{P})$ is composed of the partial profit for $(Z', \mathcal{P}')$ plus the profit achieved at $u$: half of the profit of each request in $Z$ having $u$ as an endpoint, which corresponds exactly to the difference between the profits of a partial solution at $t'$ and one at $t$. Similarly, the number of edges in each $\zeta(z)$ is exactly the number of edges in $\zeta(z')$ plus the number of ``new'' edges in $\present(t)\setminus \present(t')$ through which $z$ is routed in the partial solution, and this is reflected in our construction by setting $\theta_z = \theta'_z + \delta_u(z)$ if  $z\in Z'$, and $\theta_z = \delta_u(z)$ otherwise.

The above establishes that for every tuple that corresponds to a partial solution at $t$ (i.e., every tuple that ``should be'' in $R(t)$) will occur as a candidate tuple in our construction. To complete the proof of correctness, it still remains to argue that every candidate tuple $(\Lambda, \Theta, \Omega)$ is valid in the sense of corresponding to a partial solution $\zeta$. By induction, we know that the entry $(\Lambda', \Theta', \Omega')\in R(t')$ chosen in the branch leading to $(\Lambda, \Theta, \Omega)$ does correspond to a partial solution $\zeta'$ at $t'$. 

Let $Z''=\Lambda(t)\cup \Lambda(t')$ be the set of all tasks which were active in $\zeta'$ plus those which we assumed to be activated in $\Lambda$.
For each such $z\in Z$, let us construct $\zeta(z)$ by combining $(\Lambda, \Theta, \Omega)$ with $\zeta'(z)$. 
\begin{itemize}[topsep=3pt]
\item First, if $z\notin \Lambda(t)$, $\zeta(z)=\zeta'(z)$ since the routing is completely in the past in this case. 
\item Second, if $z\in \Lambda(t)$ but $z\notin \Lambda(t')$ (i.e., $z$ was not active in $\zeta'$), let $U_z=\{v\in \chi(t), z\in \Lambda(uv)\}.$ Because of the check on $\delta_u(z)$, $U_z$ is a singleton if and only if $z$ has $u$ as an endpoint, and a pair otherwise. If $U_z=\{v\}$, we define $\zeta(z)=\{uv\}$, and if $U_z=\{v,w\}$, we set $\zeta(z)=\{(w, u, v)\}$, i.e., the set containing the $3$-vertex path between $v$ and $w$ passing through $u$. 
\item In the last case, if $z\in \Lambda(t)\cap \Lambda'(t)$, we construct $U_z$ in precisely the same way as before, and distinguish based on its cardinality. 
If $U_z=\emptyset$, then the routing remains the same and we set $\zeta(z)=\zeta'(z)$. 
The case where 
$U_z=\{v\}$ can only happen if either $z$ has an endpoint in $u$ (in which case we simply add to $\zeta(z)$ a new path $(u,v)$), or $\zeta'(z)$ contained a path ending in $u$ (in which case we extend that path by adding the edge $uv$). 
 The final subcase is where $U_z=\{v,w\}$, and here we can correctly construct $\zeta(z):=\zeta'(z) \cup \{(v,u,w)\}$, where if $v$ and/or $w$ occurs in any of the paths in $\zeta'(z)$ then we merge these at $v$ and/or $w$.
\end{itemize}

Now, let us assume that the $\zeta$ constructed in the above way is \emph{not} a partial solution at $t$. The only way this can happen is if a task $z$ routed through $u$ is now mapped to a set containing a cycle $\iota$ (as opposed to a set of vertex-disjoint paths); in other words, one considers the case where a path in $\zeta'$ is closed into a cycle via $u$. Then by the definition of partial solutions and the fact that $\zeta'$ is a partial solution, no such $\iota$ can contain any endpoints of $z$. Thus, there will also exist a separate record $(\Lambda^*, \Theta^*, \Omega^*)\in R(t)$ which corresponds to the partial solution obtained by removing all such cycles $\iota$ (for all tasks $z$ forming such cycles) from $\zeta$. Moreover, $(\Lambda^*, \Theta^*, \Omega^*)$ supersedes $(\Lambda, \Theta, \Omega)$ because the neighbors of $u$ in $\chi(t)$ occurring in such cycles are jump-vertices. Thus, the final step of computing $R(t)$ in our algorithm prunes away all records which do not correspond to partial solutions at $t$.

By the above, there is a one-to-one correspondence between the candidate tuples constructed at $t$ and the partial solutions at $t$. Hence, the set of tuples constructed at $t$ by applying the optimality conditions is precisely the set $R(t)$, as claimed.

\smallskip
\noindent \textbf{$t$ is a join node.}\quad
Let $t_1$ and $t_2$ be the two children of $t$, and observe that $\present(t)$ is precisely the disjoint union of $\present(t_1)$ and $\present(t_2)$.

\paragraph*{Construction.}
We construct a set of candidate tuples for $t$ by enumerating over each choice of a tuple $(\Lambda_1,\Theta_1,\Omega_1)\in R(t_1)$ and $(\Lambda_2,\Theta_2,\Omega_2)\in R(t_2)$, and for each such pair we construct a tuple $(\Lambda,\Theta,\Omega)$ as follows. We set $\Lambda=\Lambda_1\cup \Lambda_2$ (i.e., each edge in $\present(t)$ maintains their mapping from the respective child), $\forall z\in \Lambda(t) \Theta(z)=\Theta_1(z)+\Theta_2(z)$ (whereas if either of the latter terms is undefined, we treat it as $0$), and $\Omega=\Omega_1+\Omega_2$. Before proceeding, we discard degenerate tuples where $\Lambda$ maps a task $z$ to more than two present edges incident to some non-endpoint of $z$ or more than one present edge incident to an endpoint of~$z$.

After completing the construction of all candidate tuples as per the above procedure, we once again compute $R(t)$ by discarding every candidate tuple that does not satisfy Optimality Conditions~(i) and~(ii), i.e., tuples which are superseded by another candidate tuple or achieve suboptimal partial profit. 

\paragraph*{Correctness.}\quad
Similarly as in the case of forget nodes, every partial solution $\zeta$ at $t$ induces a unique partial solution $\zeta_1$ at $t_1$ and a unique partial solution $\zeta_2$ at $t_2$.
        Analogously, every partial solution $\zeta$ at $t$ can be obtained by merging two partial solutions $\zeta_1$ at $t_1$ and $\zeta_2$ at $t_2$, and the choice of the latter two solutions is unique. By the correctness of our records for $t_1$ and $t_2$, the algorithm will identify a branch containing the tuple $(\Lambda_1,\Theta_1,\Omega_1)\in R(t_1)$ corresponding to $\zeta_1$ and the tuple $(\Lambda_2,\Theta_2,\Omega_2)\in R(t_2)$ corresponding to $\zeta_2$. The tuple $(\Lambda, \Theta, \Omega)$ constructed in this branch will correspond to $\zeta$, and will thus be admissible.

The above argument establishes that every admissible tuple in $R(t)$ will be present in our candidate set. However, it is important to note that---unlike in the case of forget nodes---the converse \textbf{need not hold for join nodes}. Indeed, combining two partial solutions may create closed cycles and thus not result in a partial solution (see Figure~\ref{fig:cycle}). While one could detect and deal with this issue by additionally keeping track of the order in which the present edges are visited in our records (resulting in an additional logarithmic factor in the exponent), our use of the optimality conditions allows us to circumvent it while maintaining a single-exponential running time dependency on our parameters.

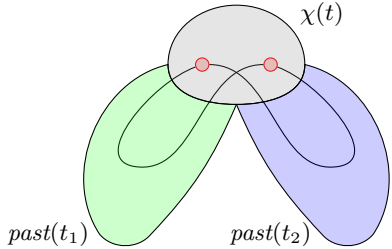
\begin{figure}[ht]
    \vspace{-0.3cm}
\begin{minipage}[c]{0.45\textwidth}
    \begin{tikzpicture}[scale=0.45,
      mycircle/.style={
         circle,
         draw=black,
         fill=black,
         fill opacity = 1,
         text opacity=1,
         inner sep=0pt,
         minimum size=5pt,
         font=\small},
         mycircle2/.style={
         circle,
         draw=white,
         fill=white,
         fill opacity = 0,
         text opacity=1,
         inner sep=0pt,
         minimum size=0pt,
         font=\small},
         mycirclered/.style={
         circle,
         draw=red,
         fill=red,
         fill opacity = 0.2,
         text opacity=1,
         inner sep=0pt,
         minimum size=5pt,
         font=\small},
         mycircle3/.style={
         circle,
         draw=black,
         fill=white,
         fill opacity = 0,
         text opacity=1,
         inner sep=0pt,
         minimum size=8pt,
         font=\small},
      myarrow/.style={},
      node distance=0.6cm and 1.2cm
      ]
    \node at (2.5,1.5) {$\chi(t)$};
    \node at (-5.5,-5) {$\past(t_1)$};
    \node at (1,-5) {$\past(t_2)$};

    \draw (1,0) node[mycirclered, draw] (v1) {};
    \draw (-1,0) node[mycirclered, draw] (v2) {};

    \draw (v1) to [in = 10, out = -20] (3,-3) to [out = 190, in = 0] (v2);
    \draw (v2) to [in = 170, out = -160] (-3,-3) to [out = -10, in = 180] (v1);
    \filldraw[fill = gray, draw = black, fill opacity = 0.2] (0,-1.175) to [out = 0, in = -90] (2,0) to [out = 90, in = 0] (0,1.75) [out = 180, in = 90] to (-2,0) to [out = -90, in = 180] (0,-1.175);
    \filldraw[fill = green, draw = black, fill opacity = 0.2] (-2,0) to [out = -90, in = 180] (0,-1.175) to [bend left = 10] (-2.5, -5) to [bend left = 50] (-4,-5) to [bend left = 55] cycle;
    \filldraw[fill = blue, draw = black, fill opacity = 0.2] (0,-1.175) to [out = 0, in = -90] (2,0) to [bend left = 55] (4, -5) to [bend left = 50] (2.5,-5) to [bend left = 10] cycle;
    \end{tikzpicture}
    \end{minipage}\hfill
    \begin{minipage}[c]{0.5\textwidth}
\caption{A join node $t$ with two children $t_1$, $t_2$, where two paths between the red vertices---one in each of the partial solutions for $t_1$ and $t_2$---create a cycle. By combining a pair of tuples $(\Lambda_1,\Theta_1,\Omega_1)\in R(t_1)$ and $(\Lambda_2,\Theta_2,\Omega_2)\in R(t_2)$ which represent partial solutions containing such paths, the algorithm could produce a tuple which does not represent any partial solution at $t$.}
\label{fig:cycle}
\end{minipage}
    \vspace{-0.3cm}
\end{figure}

Indeed, let us consider a tuple $(\Lambda, \Theta, \Omega)$ constructed by the algorithm via a combination of the tuples $(\Lambda_1,\Theta_1,\Omega_1)\in R(t_1)$ and $(\Lambda_2,\Theta_2,\Omega_2)\in R(t_2)$. By the correctness of the records at $t_1$ and $t_2$, the latter two records correspond to some partial solutions $\zeta_1$ and $\zeta_2$, respectively. 

Assume that combining $\zeta_1$ and $\zeta_2$ results in at least one task $z\in \Lambda_1(t_1)\cap \Lambda_2(t_2)$ such that there exists a set $P_{\cyc}\subseteq \zeta_1(z)\cup \zeta_2(z)$ of paths which form a cycle in $G$. Then $(\Lambda, \Theta, \Omega)$ is a candidate tuple, even though there is no corresponding partial solution. However, there also exist partial solutions $\zeta_1'$ and $\zeta_2'$ which do not contain the paths forming $P_{\cyc}$ but are otherwise identical to $\zeta_1$ and $\zeta_2$, respectively; indeed, each of the paths contributing to $P_{\cyc}$ starts and ends at $\present(t)$ and hence is ``optional''. Let $(\Lambda'_1,\Theta'_1,\Omega_1)\in R(t_1)$ and $(\Lambda'_2,\Theta'_2,\Omega_2)\in R(t_2)$ be the two tuples corresponding to $\zeta_1$ and $\zeta_2$, respectively, and let $(\Lambda',\Theta',\Omega)$ be the tuple constructed by their combination. Crucially, while neither of the aforementioned records in the children are superseded by each other, the record $(\Lambda, \Theta, \Omega)$ is immediately superseded by $(\Lambda',\Theta',\Omega)$---indeed, the two are identical except for the latter not containing the cycle $P_{\cyc}$ in the image of $z$. Hence the algorithm will in this case never add $(\Lambda, \Theta, \Omega)$ to $R(t)$.

Next, assume the converse: combining $\zeta_1$ and $\zeta_2$ does not result in any task being mapped to a cycle in $G$. In this case, the combination of the two partial solutions either results in a partial solution $\zeta$ at $t$, or contains a task $z$ which is routed through some vertex $v\in\chi(t)$ too many times (i.e., either $z$ starts at $v$ and is routed to both $\zeta_1$ and $\zeta_2$, or $z$ does not start at $v$ but occurs on at least three edges incident to $v$). In the latter case, the algorithm correctly discards the tuple $(\Lambda, \Theta, \Omega)$ due to the check carried out at the end of the first paragraph of the join node procedure. In the former case, the algorithm adds $(\Lambda, \Theta, \Omega)$ to its set of candidate tuples. 

In particular, the above implies that for every tuple that belongs to $R(t)$, our algorithm will add it to the set of candidate tuples. Moreover, the only candidate tuples which are identified by the algorithm without being admissible are superseded by some admissible candidate tuple. Hence, the set of tuples constructed at $t$ by the algorithm is precisely $R(t)$, as required.

\smallskip
\noindent \textbf{Running time.}\quad
As mentioned earlier, the size of $R(t)$ is upper-bounded by $|\III|^{\bigoh(k\cdot \Delta\cdot c)}$. The computational cost of computing the set of records at each node $t$ of $\mathcal{D}$ is upper-bounded by $|R(t)|^2$, and the number of times we need to perform this computation is upper-bounded by $|\III|$. The cost of invoking Proposition~\ref{fact:findtw} to compute $\mathcal{D}$ is dominated by the aforementioned terms, and hence we can upper-bound the running time of the algorithm by $|\III|^{\bigoh(k\cdot \Delta\cdot c)}$. \qed
\end{proof}

Recalling the discussion in Section~\ref{sec:intro}, we note that dropping any of the three parameters listed in Theorem~\ref{thm:XPctwd} results in para\NP-hardness: for $c$, $\tw$ and $\Delta$ this follows by the para\NP-hardness of \UF\ on paths~\cite{BonsmaSW11}, grids with edge-capacity $1$~\cite{Marx04} and $K_{3,n}$ with edge-capacity $1$~\cite{FleszarMS18}, respectively. 
 Hence, the main complexity-theoretic question that remains for the aforementioned three parameters is to exclude the existence of a fixed-parameter algorithm (under standard complexity-theoretic assumptions). We do so by establishing a stronger claim: \UF\ is \Wh\ parameterized by $c$ even in the well-studied setting~\cite{ChrobakWMX12,Anagnostopoulos18a,0001MW22} where the input graph is a path. We remark that all of the aforementioned lower-bound results (including Theorem~\ref{thm:Wctwd}) hold even if $\ell=|V|$, i.e., when the additional constraint on the length of the paths is ignored.

\Wctwd*

\begin{proof}
We reduce from the classical \Wh\ \MC problem (\textsc{MCC}) parameterized by the number of colors~\cite{DowneyF13,CyganFKLMPPS15}.

 Let $G=(V=(V_1,\dots,V_k),E)$ be an instance of \textsc{MCC}, and assume without loss of generality that $k$ is even; indeed, if $k$ were odd then we could produce an equivalent instance of \textsc{MCC} by adding a new part $V_{k+1}$ containing a single vertex that is adjacent to all other vertices.
    Let $(v_1, \dots, v_n)$ be an ordering on $V$ which respects the partitioning of vertices into the color classes of $G$---specifically, we require there to exist $i_0=0, i_1, \dots, i_k$ with $V_j = \{v_{i_{j-1}+1}, \dots v_{i_j}\}$ for all $j\in [k]$.

We construct an instance $\III$ of $\UF$ as follows.
    First, we create a path $P = (s_1, \dots, s_k, v_{1s}, v_{1t}, \dots, v_{ns}, v_{nt}, t_1, \dots, t_k)$, consisting of $2k$ control vertices at the start and end (forming the subpaths $s_1,\dots, s_k$ and $t_1,\dots,t_k$) and $2$ consecutive vertices of the form $v_{js}, v_{jt}$ for each $v_j\in V$ which appear in the order constructed above. We denote the corresponding partition of $\{v_{jr} | 1\leq j \leq n, r \in \{s,t\}\}$ into color classes as $(V_1', \dots, V_k')$. 

We set the capacity of each of the edges between control vertices as follows: $c(s_i, s_{i+1}) = c(t_{k-i}, t_{k-i+1}) := \sum_{j = 1}^i \max(0,k-2j+1)$. For all other edges, we set the capacity to $p^2$ where $p=0.5k$.
   We set $\tau$ to equal the sum of the capacities of all edges in $P$, and set $\ell:=|V(P)|$ (or to any other sufficiently large number).
    Finally, we create the following tasks:
         \begin{itemize}[topsep=3pt]
        \item $T_{sj} = (s_i, v_{js},\max(0,k-2i+1), \bullet)$ for all $v_j \in V_i$,
        \item $T_i = (v_{is}, v_{it}, \max(k-i, i-1), \bullet)$ for all $v_i \in V$,
        \item $T_{tj} = (v_{jt}, t_i, \max(0,2i-k-1), \bullet)$ for all $v_j \in V_i$,
        \item $T_{ij} = (v_{it}, v_{js}, 1, 1)$ for all $v_iv_j \in E$ with $i < j$ and $c(v_i) \neq c(v_j)$,
    \end{itemize}
    where the profit of each task is set to its demand times the number of edges between its two endpoints. Observe that since $\tau$ is the sum of all edge capacities in $P$, the instance $\III$ constructed above is a yes-instance if and only if there exists a set of tasks which together use up all of the capacity of all edges in $P$; indeed, each task provides precisely as much profit as the amount of capacity it takes from $\III$. We call tasks of the form $T_{ij}$ \emph{edge tasks}.
    
    Intuitively, the reduction works by encoding a solution to \textsc{MCC} via a collection of tasks which have endpoints in precisely one consecutive pair of vertices in each $V'_i$, $i\in [k]$, as illustrated in Figure~\ref{fig:reduction}.
    To complete the proof, it suffices to formally establish that $G$ contains a multicolored clique if and only if there is a flow on $P$ with total profit $\tau$. 
         \begin{figure}[ht]
  \begin{center}
    \begin{tikzpicture}[
      mycircle/.style={
         circle,
         draw=black,
         fill=black,
         fill opacity = 1,
         text opacity=1,
         inner sep=0pt,
         minimum size=5pt,
         font=\small},
         mycircle2/.style={
         circle,
         draw=white,
         fill=white,
         fill opacity = 0,
         text opacity=1,
         inner sep=0pt,
         minimum size=0pt,
         font=\small},
         mycircle3/.style={
         circle,
         draw=black,
         fill=white,
         fill opacity = 0,
         text opacity=1,
         inner sep=0pt,
         minimum size=5pt,
         font=\small},
         mycircleblue/.style={
         circle,
         draw=blue,
         fill=blue,
         fill opacity = 0.2,
         text opacity=1,
         inner sep=0pt,
         minimum size=5pt,
         font=\small},
         mycirclered/.style={
         circle,
         draw=red,
         fill=red,
         fill opacity = 0.2,
         text opacity=1,
         inner sep=0pt,
         minimum size=5pt,
         font=\small},
         mycircleorange/.style={
         circle,
         draw=orange,
         fill=orange,
         fill opacity = 0.2,
         text opacity=1,
         inner sep=0pt,
         minimum size=5pt,
         font=\small},
         mycirclecyan/.style={
         circle,
         draw=cyan,
         fill=cyan,
         fill opacity = 0.2,
         text opacity=1,
         inner sep=0pt,
         minimum size=5pt,
         font=\small}
      ]
     \draw (-1,0) node[mycircleblue, draw, label=right:$v_1$] (v1) {};
    \draw (1,0) node[mycircleblue, draw,label=left:$v_2$] (v2) {};
    
    \draw (0,-1) node[mycirclered, draw,label=right:$v_3$] (v3) {};
    
    \draw (-1,-2) node[mycircleorange, draw,label=right:$v_4$] (v4) {};
    \draw (1,-2) node[mycircleorange, draw,label=left:$v_5$] (v5) {};
    
    \draw (0,-3) node[mycirclecyan, draw,label=right:$v_6$] (v6) {};
    
    \draw (v1) edge[very thick] (v3);
    \draw (v2) to (v3);
    \draw (v3) edge[bend left = 25, very thick] (v6);
    \draw (v1) edge[very thick] (v4);
    \draw (v4) edge[very thick] (v6);
    \draw (v1) edge[very thick] (v6);
    \draw (v3) edge[very thick] (v4);
    \draw (v2) to (v5);
    \draw (v5) to (v3);

         \draw (2.1,-1.5) node[mycircle3,draw] (s1) {};
    \draw (2.4,-1.5) node[mycircle3,draw] (s2) {};
    \draw (2.7,-1.5) node[mycircle3,draw] (s3) {};
    \draw (3,-1.5) node[mycircle3,draw] (s4) {};

    \draw (3.5,-1.5) node[mycircleblue, draw] (v1s) {};
    \draw (4,-1.5) node[mycircleblue, draw] (v1t) {};
    
    \draw (4.5,-1.5) node[mycircleblue, draw] (v2s) {};
    \draw (5,-1.5) node[mycircleblue, draw] (v2t) {};
    
    \draw (5.5,-1.5) node[mycirclered, draw] (v3s) {};
    \draw (6,-1.5) node[mycirclered, draw] (v3t) {};
    
    \draw (6.5,-1.5) node[mycircleorange, draw] (v4s) {};
    \draw (7,-1.5) node[mycircleorange, draw] (v4t) {};
    
    \draw (7.5,-1.5) node[mycircleorange, draw] (v5s) {};
    \draw (8,-1.5) node[mycircleorange, draw] (v5t) {};
    
    \draw (8.5,-1.5) node[mycirclecyan, draw] (v6s) {};
    \draw (9,-1.5) node[mycirclecyan, draw] (v6t) {};

    \draw (9.5,-1.5) node[mycircle3,draw] (t1) {};
    \draw (9.8,-1.5) node[mycircle3,draw] (t2) {};
    \draw (10.1,-1.5) node[mycircle3,draw] (t3) {};
    \draw (10.4,-1.5) node[mycircle3,draw] (t4) {};

    \draw (s1) -- (s2) -- (s3) -- (s4) -- (v1s) -- (v1t) -- (v2s) -- (v2t) -- (v3s) -- (v3t) -- (v4s) -- (v4t) -- (v5s) -- (v5t) -- (v6s) -- (v6t) -- (t1) -- (t2) -- (t3) -- (t4);
    
    \draw (v1s) -- (v1t) node [midway, label={[label distance=-2pt]above:$v_1$}] {};
     \draw (v3s) -- (v3t) node [midway, label={[label distance=-2pt]above:$v_3$}] {};
    \draw (v4s) -- (v4t) node [midway, label={[label distance=-2pt]above:$v_4$}] {};        
     \draw (v6s) -- (v6t) node [midway, label={[label distance=-2pt]above:$v_6$}] {};    
   
    \draw (s1) edge[bend right = 30, very thick] (v1s);
    \draw (s2) edge[bend left = 30, very thick] (v3s);

    \draw (v4t) edge[bend left = 30, very thick] (t3);
    \draw (v6t) edge[bend right = 30, very thick] (t4);

    \draw (v1t) edge[bend right = 30, very thick] (v3s);	
    \draw (v1t) edge[bend right = 30, very thick] (v4s);
    \draw (v1t) edge[bend right = 30, very thick] (v6s);
    \draw (v3t) edge[bend left = 30, very thick] (v4s);
    \draw (v3t) edge[bend right = 30, very thick] (v6s);
    \draw (v4t) edge[bend right = 30, very thick] (v6s);

    \draw (v1s) edge[bend left = 30, very thick] (v1t);
    \draw (v3s) edge[bend left = 30, very thick] (v3t);
    \draw (v4s) edge[bend left = 30, very thick] (v4t);
    \draw (v6s) edge[bend left = 30, very thick] (v6t);
    \end{tikzpicture}
\end{center} 
\caption{An illustration of the construction used in the proof of Theorem~\ref{thm:Wctwd}. On the left, we depict an instance of \MC\ with a solution highlighted in bold. On the right, we show the corresponding instance of \UF\ on the path (the horizontal edges), where the tasks corresponding to the highlighted clique are depicted via bold curves. For ease of presentation, we omit the tasks not chosen in this particular solution from the figure.}
\label{fig:reduction}
\end{figure}
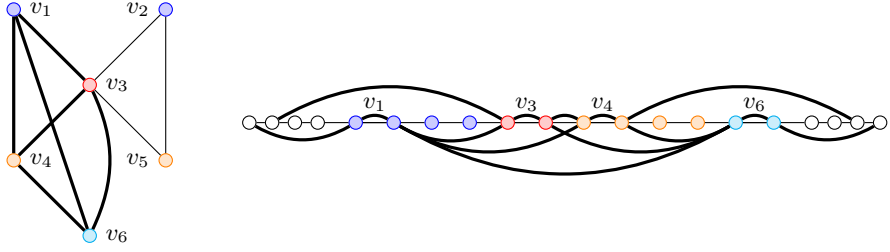

    For the forward direction, let $C=(v_{\alpha(1)}, \dots, v_{\alpha(k)}) \subseteq V$ be a multicolored clique. We claim that the set of tasks $\{T_{s\alpha(i)}, T_{\alpha(i)}, T_{t\alpha(i)}, T_{\alpha(i) \alpha(i')}~|~1\leq i < i'\leq k\}$ represents an unsplittable flow that fully exhausts all capacities on $P$. 
         
    As $C$ is a multicolored clique, we choose exactly one vertex of every color, so by construction the capacity of all edges of the form $s_i s_{i+1}$ and $t_i t_{i+1}$ is fully exhausted. Let $vw$ be an edge on $P$ between $v_{\alpha(i)t}$ and $v_{\alpha(i+1)s}$, and note that the tasks which are routed through $vw$ are $T_{sq}$ for $q \geq i+1$, $T_{tq}$ for $q \leq i$ and $i \cdot (k-i)$ different tasks of the form $T_{qx}$ for $q\leq i$, $x\geq i+1$. 
    To count how much capacity is used by the tasks routed through $vw$, we consider $vw$ compared to its adjacent egde $uv$. If $u,v,w \notin C$ then $uv$ and $vw$ have the same active tasks. If $u = v_{\alpha(i) s}$ and $v = v_{\alpha(i) t}$, then $uv$ gains the task $T_{\alpha(i)}$ but loses $T_{t\alpha(i)}$ and all edge tasks connecting $v_{\alpha(i) t}$ to the $k-i$ later vertices. So the total difference in used capacity is
    \begin{eqnarray*}
        &&\max(k-i, i-1) - (k-i) - \max(0,2i-k-1) \\
        &=& \max(0, 2i-k-1) - \max(0, 2i-k-i) = 0
    \end{eqnarray*}
    Similarly, if $v = v_{\alpha({i+1}) s}$ and $w = v_{\alpha({i+1}) t}$ the difference is also $0$. As all these edges have the same used capacity, we can simply check the capacity used by $(s_k, v_{1s})$, which is:
    \begin{eqnarray*}
        \sum_{i=1}^k \max(0, k-2i+1) =              \sum_{i=0}^{p-1} 2p+1 = p^2                       \end{eqnarray*}
    So the entire capacity of every edge is exhausted, so there exists a maximal flow using all the capacity on $P$, and thus a set of tasks giving a total profit of $B$. 
    \\

    For the converse direction, let $T_{\emph{max}}$ be a set of tasks of total profit $\tau$. In this direction, we assume w.l.o.g.\ that $T_{\emph{max}}$ contains no tasks with profit $0$, as these will not increase the total profit\footnote{To avoid confusion, we remark that the solutions constructed for the forward direction included superfluous tasks of capacity $0$.} Then, by construction, for each color $j\in [k]$ the considered solution $T_{\emph{max}}$ of $\UF$ contains precisely one task from $\{T_{sj}, T_{tj}\}$, as the other has a demand and profit of $0$. 
    
    Next, for each $j\in [k]$ let $\beta(j)$ denote the index of the vertex $v_{\beta(j)}$ in $V_j$ such that there exists either a task $T_{s\beta(j)}$ or $T_{t\beta(j)}$ in $T_{\emph{max}}$. Note that such a $\beta(j)$ must exist for each $j\in [k]$, because $k$ is even and thus either $\max (0, k-2j+1)$ or $\max (0, 2j-k-1)$ is greater than $0$. Our aim is now to establish the following claim: for each $j < \frac{k+1}{2}$, $T_{\emph{max}}$ contains
    \begin{itemize}[topsep=3pt]
    \item for each $i_1 < i_2 \leq j$, the tasks $T_{\beta(i_1)\beta(i_2)}$ and $T_{\beta(k-i_2+1)\beta(k-i_1+1)}$, and
    \item for each $i\leq j$:
    \begin{itemize}[topsep=3pt] 
    \item exactly $k-i$ edge tasks starting at $v_{\beta(i)t}$,
    \item exactly $k-i$ edge tasks ending at $v_{\beta(k-i+1)t}$, 
    \item the tasks $T_{\beta(i)}$, $T_{\beta(k-i+1)}$, and
    \item no other tasks starting or ending in $V_i'$ or $V_{k-i+1}'$ besides $T_{\beta(i) s}$, $T_{\beta(k-i+1) t}$.
    \end{itemize}
    \end{itemize}
  
We establish the aforementioned claim by inducrion on $j$.
    By construction, the entire capacity in $P$ until $v_{\beta(1) s}$ is exhausted. As $T_{\beta(1) s}$ ends in $v_{\beta(1) s}$, the edge $v_{\beta(1) s}v_{\beta(1) t}$ is not fully exhausted and thus $T_{\emph{max}}$ must contain $T_{\beta(1)}$. To fully exhaust the following edges, $T_{\emph{max}}$ must contain $k-1$ edge tasks starting at $v_{\beta(1) t}$, as no other tasks of positive demand start at $v_{\beta(1) t}$. By our construction, no edge tasks start and end in the same color, thus no other task starts or ends in $V_1'$. By the same arguments $T_{\emph{max}}$ contains $T_{\beta(k) t}$ and exactly $k-1$ edge tasks ending in $v_{\beta(k)}$. Thus, the claim holds for $j=1$.

    Now let us assume that until some choice of $j$, the induction hypothesis holds and $j+1 < \frac{k+1}{2}$. Then let $v_{m s}$ be the first vertex after $v_{\beta(j) t}$ such that $T_{\emph{max}}$ contains an edge task ending in $v_{m s}$. Note that no other edge task can start between $v_{\beta(j) t}$ and $v_{m s}$ as the capacity is fully used up. As the capacity of the edge $v_{m s}v_{m t}$ is not fully exhausted, $T_{\emph{max}}$ must contain $T_m$. 
     If $m > \beta(j+1)$, then the edge $v_{\beta(j+1) s}v_{\beta(j+1) t}$ is not fully exhausted, as we cannot take $T_{\beta(j+1)}$ into $T_{\emph{max}}$ without exceeding the capacity due to $\max (0, k-2(j+1)+1) < \max (k-(j+1), j)$, which contradicts our assumption. 
    Likewise, if $m < \beta(j+1)$, then $v_{m s} v_{m t}$ is not fully exhausted, as at most $j$ edge tasks can end in $v_{m s}$, but $T_m$ has a capacity of $\max (k-j-1, j) = k-j-1$ as $j+1 < \frac{k+1}{2}$. Thus $m = \beta(j+1)$ and exactly $j$ edge tasks must end in $v_{\beta(j+1) s}$---one from each $v_{\beta(i) t}$ with $i \leq j$---to be able to include $T_{\beta(j+1)}$ in $T_{\emph{max}}$. 
     
    As, starting at $v_{\beta(j+1) t}$, the entire capacity is used along all edges of the path and no non-edge task in $T_{\emph{max}}$ ends in $V_{j+1}'$ after $v_{\beta(j+1) t}$, for another taks to start in $V_{j+1}'$ some edge task would need to end in some $v_{m s} \in V_{j+1}'$. But at most $j$ edge tasks can end in $v_{m s}$, and $\max (k-j-1,j) = k-j-1$, the edge $v_{m s} v_{m t}$ would not be fully used. So no other task starts or ends in $V_{j+1}'$. By mirroring each of these steps we also get the same result for $v_{\beta(k-j)}$ and $V_{k-j}'$. 
     
    As $k = 2p$, the above shows that the existence of $T_\emph{max}$ implies the existence of a clique in $G$. Indeed, our induction result covers the whole graph, exactly $k-j$ edge tasks start at $v_{\beta(j) t}$, and for $i > j$ there can be at most one edge from $v_{\beta(j) t}$ to $v_{\beta(i) s}$ and to no other vertex in $V_i'$, and there are exactly $k-j$ different sets $V_i'$ with $i> j$. In particular, $\{v_{\beta(1)}, \dots, v_{\beta(k)}\}$ forms a multicolored clique, concluding the proof of the converse direction. 
    \qed
\end{proof}

 \section{The Missing Ingredient for Fixed-Parameter Tractability}
\label{sec:bounded_l}
Our results in Section~\ref{sec:unbounded_l} show that not even parameterizing by a combination of the treewidth \tw, maximum degree $\Delta$ and maximum capacity $c$ suffice to achieve fixed-parameter tractability for \UF. At this point, one may wonder whether it would be possible to achieve fixed-parameter tractability by strengthening the structural restrictions on the input graph---in particular by replacing treewidth with a different natural graph parameter such as 
treedepth~\cite{sparsity}, the vertex cover number~\cite{KorhonenP15,BhoreGMN20,BodlaenderGP23}, treecut-width~\cite{MarxWollan14,Wollan15,GanianKS22}, and the feedback edge number~\cite{ChaplickGFGRS22,KoanaFN23,BGR24}. Recalling that in Section~\ref{sec:intro} we already ruled out any sort of tractability for \UF\ under such parameterizations \emph{alone}, below we also rule out fixed-parameter tractability when they are combined with $c$:

\begin{itemize}[topsep=3pt]
\item For treedepth and the vertex cover number, it is futile to combine these parameters with $\Delta$ as this also bounds the size of the graph. On the other hand, \UF\ parameterized by $c$ plus either of these parameters remains \NP-hard due to the previously-mentioned \NP-hardness on $K_{3,n}$~\cite{FleszarMS18}.
\item When parameterizing by treecut-width or the feedback edge number in combination with $c$ and $\Delta$, we can rule out fixed-parameter tractability directly via the established \Whness\ on paths (Theorem~\ref{thm:Wctwd}). The same also holds for the recently introduced ``slim'' variants of tree-cut width~\cite{GanianK22,GanianHKOS22}.
\end{itemize}

But while it seems essentially impossible to achieve fixed-parameter tractability via placing stronger restrictions on the input graph, as the main contribution of this section we show that adding the bound $\ell$ on the length of admissible flow routes into the parameterization allows us to circumvent the \Whness\ of the problem. Indeed, our aim is to prove:

\begin{lemma}
\label{lem:fptruntime}
\UF\ can be solved in time $|\III| \cdot (c+1)^{\bigoh(\tw^2 \cdot \Delta^{2 \ell})}$.
\end{lemma}

Before proceeding to the proof of Lemma~\ref{lem:fptruntime}, we observe that it immediately implies the two remaining algorithmic upper bounds depicted in Figure~\ref{fig:landscape}:

\fptctwdell*

\xptwdell*
   
\begin{proof}[Proof of Lemma~\ref{lem:fptruntime}] 
We begin by invoking Proposition~\ref{fact:findtw} to compute a nice tree-decomposition $(D,\chi)$ of $G$ of width $k$, where $k\leq 2\tw+1$.
Our proof utilizes a leaf-to-root dynamic programming algorithm working on $D$, which for each node computes partial solutions for the problem---however the records and the procedures of this dynamic program are entirely different from those used in the proof of Theorem~\ref{thm:XPctwd}.

For a node $t\in V(D)$, we define the set $\vis(t)$ of \emph{visible vertices} as the set containing all vertices of $G$ that have distance at most $\ell$ from $\chi(t)$. The set $E_{\vis}(t)$  of \emph{visible edges} then contains all edges of $G$ whose both endpoints are visible vertices. We note that for each node $t$, $|\vis(t)|\leq |\chi(t)| \cdot \Delta^\ell \leq \bigoh(\tw \cdot \Delta^\ell)$, and for visible edges we analogously obtain $|E_{\vis}(t)|\leq \bigoh(\tw^2 \cdot \Delta^{2\ell})$.

To better keep track of the different tasks at a node $t$, we will partition them into two categories: a task is \emph{active} if both of its endpoints lie in $\past(t)$, and \emph{future} if it has at least one endpoint in $\future(t)$.
A \emph{partial solution} at a node $t$ is a mapping $\zeta$ from a subset $Z$ of active tasks to paths of length at most $\ell$ in $G$ such that, for any edge $e$, its capacity is high enough to accomodate the flow routed through it in the partial solution, i.e., $\mu_\zeta(e):=\sum_{z\in Z, e\in \zeta(z)} d(z) \leq \gamma(e)$.
The \emph{profit} of $\zeta$ is then simply the sum of the profits of tasks in $Z$, i.e., $\sum_{z\in Z} w(z)$.

The syntax of our records is defined as follows: for each node $t$, $R(t)$ is a set of tuples of the form $(\Lambda, \Omega)$  where $\Lambda: E_{\vis}(t)\rightarrow [c]_0$ and $\Omega\in \N$.
Intuitively, $\Lambda$ captures the capacity that has already been consumed by routed active tasks on each visible edge, while $\Omega$ is the highest profit that can be achieved for a given $\Lambda$. Formally, we define the semantics of $R(t)$ by stipulating that $(\Lambda, \Omega)\in R(t)$ if and only if there exists a partial solution $\zeta$ at $t$ with the following properties:
\begin{enumerate}[topsep=4pt]
\item $\forall e\in E_{\vis}(t), \mu_\zeta(e)=\Lambda(e)$, and
\item the profit of $\zeta$ is exactly $\Omega$, and
\item every other partial solution $\zeta'$ such that $\forall e\in E_{\vis}(t), \mu_\zeta'(e)=\Lambda(e)$ achieves a profit that is upper-bounded by $\Omega$.
\end{enumerate}

For the root node $r$ of $D$, we observe that all tasks are active and partial solutions are thus routings of an arbitrary subset of all tasks, and that $E_{\vis}(r)=\emptyset$. By the definition of our semantics, we have only a single entry $(\emptyset, \Omega)\in R(r)$ where $\Omega$ is the highest profit that can be achieved by any flow routing in $\III$. Hence if we can compute $R(r)$, it suffices to compare this value with the target $\tau$ to solve $\III$. Moreover, computing $R(t)$ for a leaf node $t$ is trivial: all tasks are future and hence $R(t)=\{(\emptyset,0)\}$. Thus, it remains to show how the algorithm can compute the records in a leaf-to-root fashion when $t$ is an introduce, forget or join node.

\smallskip
\noindent \textbf{$t$ is an introduce node.}\quad
Let $t'$ be the child of $t$ and $u\in\chi(t)\setminus \chi(t')$ be the unique vertex introduced at $t$. Notice that while $E_{\vis}(t)\supseteq E_{\vis}(t')$, the set of active tasks is the same at $t$ and at $t'$ since $\past(t)=\past(t')$, and thus the partial solutions at $t$ are exactly the partial solutions at $t'$. In all of these partial solutions, because of the bound $\ell$ on the length of the paths associated with each task whose both endpoints lie in $\past(t')$, none of the edges in $E_{\vis}(t)\setminus E_{\vis}(t')$ can be used by any partial solution at $t$.

Thus, for each tuple $(\Lambda', \Omega)\in R(t')$, we construct a new tuple $(\Lambda,\Omega)\in R(t)$ by setting $\Lambda(e):=\Lambda'(e)$ for each edge in $E_{\vis}(t')$, and $\Lambda(e):=0$ for all each edge $e\in E_{\vis}(t)\setminus E_{\vis}(t')$. The correctness of this procedure follows directly from the observations in the previous paragraph.
 
\smallskip
\noindent \textbf{$t$ is a forget node.}\quad
Let $t'$ be the child of $t$ and $u\in\chi(t')\setminus \chi(t)$ be the unique vertex forgotten at $t$. Since $\past(t)=\past(t')\cup\{u\}$, a future task at $t'$ could become active in $t$ (in particular, this occurs for each task with one endpoint at $u$ and the other in $\past(t')$). Intuitively, one could apply a brute-force procedure for each such task to determine its impact on our records, but unfortunately here the number of such ``new'' tasks $t_1, \dots, t_x$ may be very large and hence it would be entirely infeasible to apply such a procedure for all of them simultaneously. To deal with this somewhat unusual obstacle, we apply a second layer of dynamic proramming in order to compute $R(t)$.

For each $i\in [x]_0$, we say that $\zeta$ is a \emph{partial solution of rank $i$ at $t$} if $\zeta$ is a partial solution at $t$ such that its preimage $Z$ satisfies the following condition: $\forall j>i, t_j\notin Z$. Observe that the set of partial solutions of rank $0$ ($x$) at $t$ is precisely the set of ``usual'' partial solutions at $t'$ (at $t$).
We now iteratively create a set $R_i(t)$ for each $i\in [x]_0$ with the same syntax semantics as the records of $R(t)$, but only restricted to partial solutions of rank $i$. Formally, $(\Lambda, \Omega)\in R_i(t)$ if and only if there exists a partial solution \textbf{of rank $i$} $\zeta$ at $t$ with the following properties:
\begin{enumerate}[topsep=4pt]
\item $\forall e\in E_{\vis}(t), \mu_\zeta(e)=\Lambda(e)$, and
\item the profit of $\zeta$ is exactly $\Omega$, and
\item every other partial solution \textbf{of rank $i$}  $\zeta'$ such that $\forall e\in E_{\vis}(t), \mu_\zeta'(e)=\Lambda(e)$ achieves a profit that is upper-bounded by $\Omega$. \label{propertythreeoptimality}
\end{enumerate}

Notice that $R_x(t)=R(t)$, meaning that it suffices to dynamically compute all of the tables $R_i(t)$ from $i=0$ up to $i=x$. Let us begin by describing the inductive step, i.e., the computation of all entries in $R_{i+1}(t)$ from the correctly computed entries in $R_i(t)$.
Towards this end, let $t_{i+1}=(u, v, d, w)$, and note that a partial solution $S$ of rank $i+1$ can always be decomposed in two parts: $S_i$ a partial solution of rank $i$, and $R_{i+1}$ is either empty or a path (i.e., routing) for $t_{i+1}$. Moreover, if there is no partial solution of rank $i+1$ superseding $S$ --i.e. inducing the same mapping on visible edges, but achieving a better profit-- there can be no partial solution $S_i'$ of rank $i$ that supersedes $S_i$, since otherwise we could combine $S'_i$ with $R_{i+1}$ to obtain a partial solution of rank $i+1$ superseding $S$. In other words, every partial solution $\zeta$ of rank $i+1$ satisfying Property~\ref{propertythreeoptimality} can be obtained by combining some partial solution of rank $i$ satisfying Property~\ref{propertythreeoptimality} and some routing of $t_{i+1}$.

As for enumerating all routings of $t_{i+1}$, we observe that since the task ends at $u$ and can be routeed through a path of length at most $\ell$, a partial solution may only route $t_{i+1}$ through visible edges. Because of this, we only need to check on these edges that the capacities of edges are respected when combining a partial solution of rank $i$ and a routing of $t_{i+1}$, and this information is indeed stored in the records at $R_{i}(t)$.

Thus, we can correctly compute $R_{i+1}(t)$ via the following procedure:
\begin{enumerate}[topsep=4pt]
\item Branch over all paths $P$ between $u$ and $v$ of length at most $\ell$ in $G$.
\begin{itemize}[topsep=3pt]
\item For each choice of $P$, branch over all feasible entries $(\Lambda, \Omega)$ in $R_{i}(t)$.
\begin{itemize}[topsep=3pt]
\item For each choice of $(\Lambda, \Omega)$, for each edge $e$ on the path $P$, check whether $\Lambda_e+d \leq \gamma(e)$; if this fails then proceed to the next branch, and otherwise set $\Lambda^*(e):=\Lambda(e)+d$. 
\item For each edge not on $P$, set $\Lambda^*(e):=\Lambda(e)$ and then set $\Omega^*=\Omega +w$ to construct a candidate tuple $(\Lambda^*,\Omega^*)$.
\item Compute an additional candidate tuple $(\Lambda^*=\Lambda, \Omega^*=\Omega)$.
\end{itemize}
\end{itemize}
\item For each $\Lambda^*$, insert the candidate tuple $(\Lambda^*,\Omega^*)$ (if one exists) achieving minimum $\Omega^*$ into $R_{i+1}(t)$.
\end{enumerate}

It remains to describe how to compute $R_0(t)$; here, we only need to observe that the set of visible edges at $t$ is a subset of the visible edges at $t'$. Thus for each entry $(\Lambda', \Omega')\in R(t')$, we can construct $\Lambda$ as the restriction of $\Lambda'$ to $E_{vis}(t)$ and obtain a candidate tuple $(\Lambda, \Omega')$ corresponding to the same partial solution. After doing this for all entries at $t'$, we only add those candidate tuples into $R_0(t)$ which are not superseded by another tuple with the same $\Lambda$ but lower $\Omega$. The correctness of this procedure follows directly by the direct correspondence between partial solutions at $t'$ and those of rank $0$ at $t$.

\smallskip
\noindent
\textbf{$t$ is a join node.}\quad
Let $t_1$ and $t_2$ be the two children of $t$. Observe that no task could have been simultaneously active in both children of $t$, and that every task that was active in at least one child of $t$ will remain active in $t$; however, we may now have active tasks at $t$ that were not active in either of the children (simply due to having an endpoint in each of $\past(t_1)$ and $\past(t_2)$).
Let us once again denote these ``newly active'' tasks $t_1, \dots, t_x$, ordered arbitrarily.   To deal with these tasks, we will make use of the same definition of partial solutions of rank $i$ at $t$ and the same notion of $R_i(t)$ as when dealing with the previous case (i.e., with forget nodes). 

Moreover, the computation of the set $R_{i+1}(t)$ from $R_{i}(t)$ for each $i\in [x]_0$ is carried out by the same procedure as for the forget nodes as well, and correctness of each set $R_{i+1}(t)$ follows from the correctness of $R_{i}(t)$ by the exact same argument as above. 
The only difference lies in the computation of $R_0(t)$.

Here, we remark that every partial solution $S$ of rank $0$ at $t$ can be decomposed into a partial solution $S_1$ at $t_1$ and a partial solution $S_2$ at $t_2$, since the active tasks are exactly those which were active at precisely one of these two child nodes. The only part of the graph where $S_1$ and $S_2$ can overlap--i.e., can both use the same edges---is within $E_{\vis}(t)=E_{\vis}(t_1)=E_{\vis}(t_2)$. 
Crucially, if any of the aforementioned two partial solutions at the children, say $S_1$, were to be superseded by another partial solution $S_1'$ at $t_1$ in the sense of both requiring the same edge capacities in $E_{\vis}(t_1)$ but $S_1$ yielding a strictly smaller profit $\Omega_1$, then one could combine $S_1'$ with $S_2$ to obtain a partial solution of rank $0$ at $t$ which supersedes $S$. 

Thus, every entry in $R_0(t)$ can be correctly obtained via a combination of entries in $R(t_1)$ and $R(t_2)$. In particular, we construct the entries of $R_0(t)$ as follows:
for each entry $(\Lambda_1, \Omega_1)$ at $t_1$ and $(\Lambda_2, \Omega_2)$ at $t_2$, we define a candidate entry $(\Lambda,\Omega)$ as follows: $\forall e \in E_{visible}, \Lambda(e)=\Lambda_1(e)+\Lambda_2(e)$ and $\Omega=\Omega_1+\Omega_2$. We then check, for each candidate entry constructed in the above way, whether it is superseded by another candidate entry of the same $\Lambda$ but smaller $\Omega$; we discard all superseded entries and add all those that remain to $R_0(t)$

Note that the check at the end of the previous paragraph is necessary: two distinct combinations of $\Lambda_1$, $\Lambda_2$ and $\Lambda'_1$, $\Lambda'_2$ may lead to the same choice of $\Lambda$. On the other hand, the correctness of our computation of $R_0(t)$ that only considers entries in 

\smallskip
\noindent \textbf{Running time.}\quad
The number of entries in $R(t)$ is upper-bounded by $(c+1)^{|E_{\vis}(t)|}\leq(c+1)^{\bigoh(\tw^2 \cdot \Delta^{2\ell})}$, and the same bound also applies to each of the sets $R_i(t)$.
Moreover, the number of paths of length at most $\ell$ between two vertices is upper bounded by $\Delta^\ell$.
Thus, it takes time at most $(c+1)^{\bigoh(\tw^2 \cdot \Delta^{2 \ell})}\cdot \Delta^\ell= (c+1)^{\bigoh(\tw^2 \cdot \Delta^{2 \ell})}$ to compute the entries of each of our record sets.
Since each task only switches from being future to active precisely once, the number of times the entries of a new record set need to be computed is upper-bounded by the number of nodes of $D$ plus the number of tasks, which can be altogether upper-bounded by $\bigoh(|\III|)$.
 Overall, we obtain that the running time of the dynamic programming algorithm presented above can be upper-bounded by: $\bigoh(|\III|) \cdot (c+1)^{\bigoh(\tw^2 \cdot \Delta^{2 \ell})}$.
\qed
\end{proof}

As our final result, we complement Theorem~\ref{thm:xptwdell} by ruling out fixed-parameter tractability under the same parameterization.

\Wtwdell*
   
\begin{proof}
The result follows from a straightforward reduction from the \UBP  problem: given a set of $k$ bins each of capacity $m$ and a set of $p$ items each of size $s_i$ ($i\in [p]$) where $h_i$ is encoded in unary, determine whether it is possible to assign each of the items into a bin so that the total size of items assigned to each of the bins is precisely $m$. \UBP\ is known to be \Wh\ when parameterized by $k$~\cite{JansenKMS10}. \\

The reduction constructs an instance $\III$ of \UF\ as follows. We first construct a pair of pairwise non-adjacent vertices $s$ and $t$. Then for each item $i$ in the initial \UBP instance, we create the task $(s,t,{h_i},1)$.
Finally, we construct precisely $k$ additional vertices (one for each bin) and connect each such vertex to both $s$ and $t$ via edges of capacity $m$. \\

The equivalence between $\III$ and the original instance of \UBP. is immediate, since there is a direct correspondence between each $s$-$t$ path in the former and a bin in the latter. To complete the proof, we observe that the constructed graph has maximum degree $k$, treewidth $3$ and every path from $s$ to $t$ is of length $2$.\qed
\end{proof}

\section{Concluding Remarks}
While our results highlight that \UF\ remains highly intractable even in severely restricted settings, they also provide some of the first algorithmic results for the problem on general graphs. Combined, the results also paint a surprisingly rigorous and comprehensive picture of the complexity landscape of \UF. For instance, in addition to the discussion at the beginning of Section~\ref{sec:unbounded_l}, one may note that the \XP-tractability arising from Theorem~\ref{thm:XPctwd} cannot be achieved without a bound on the maximum degree of the input graph even when restricted to the class of trees, since the problem is known to remain \NP-hard even on trees of edge capacity at most $2$~\cite{GargVY97}.

A natural question left for future work would be to study the complexity of \UF\ under other parameterizations that are not ``graph-structural'' in nature, such as when parameterizing by the number of routed tasks in the solution or the target profit $\tau$; previous work has considered this question primarily on specific graph classes such as trees~\cite{Martinez-MunozW21}.

    \bibliographystyle{plainurl}
\bibliography{main.bib}

\end{document}